\documentclass[envcountsame,runningheads]{llncs}

\usepackage{amsmath}
\usepackage{amssymb}
\usepackage{bussproofs}
\usepackage{proof}
\usepackage[all]{xy}
\usepackage{adjustbox}
\usepackage{etoolbox}
\usepackage{stackrel}
\usepackage{stmaryrd}
\usepackage{xcolor}

\renewcommand{\qed}{\hfill\ensuremath{\dashv}}

\newcommand{\cmp}[1]{|#1|}
\newcommand{\ConL}{\mathsf{C}_L}
\newcommand{\ConR}{\mathsf{C}_R}
\newcommand{\crit}[1]{\mathsf{crit}(#1)}

\newcommand{\dpth}[1]{\mathsf{dp}(#1)}
\newcommand{\entails}{\vdash}
\newcommand{\EW}{\mathsf{EW}}
\newcommand{\K}{\mathsf{K}}
\newcommand{\KB}{\mathsf{KB}}
\newcommand{\Kt}{\mathsf{Kt}}
\newcommand{\HKt}{\mathsf{HKt}}

\newcommand{\LNS}{\mathsf{LNS}}

\newcommand{\mrg}{\oplus}

\newcommand{\nforces}{\not\forces}
\newcommand{\sa}{\Rightarrow}
\newcommand{\SL}{\mathsf{SL}}
\newcommand{\SR}{\mathsf{SR}}

\newcommand{\transl}[1][]{\ifstrempty{#1}{ \tau }{
    \tau \left( #1 \right)}} 
\newcommand{\W}{\mathsf{W}}

\newcommand{\wbx}{\Box}
\newcommand{\wdi}{\Diamond}
\newcommand{\bbx}{\blacksquare}
\newcommand{\bdi}{\blacklozenge}
\newcommand{\fwd}{\nearrow}
\newcommand\bwd{\swarrow}
\newcommand\upd{\updownarrow}

\newcommand{\stt}[1][]{\stackrel{#1}{\Rightarrow}}
\newcommand{\seq}[2]{#1 \stt #2}

\newcommand{\atm}{\texttt{Atm}}
\newcommand{\forces}{\Vdash}

\newcommand{\mcalG}{\mathcal{G}}
\newcommand{\mcalH}{\mathcal{H}}

\newcommand{\idrl}{(id)}

\title{Syntactic cut-elimination and backward proof-search for tense
  logic via linear nested sequents (Extended version)\thanks{Supported by WWTF project
    MA16-28}}

\author{Rajeev Gor{\'e}\inst{1} and Bj{\"o}rn Lellmann\inst{2}}

\institute{ Research School of Computer Science,
  Australian National University \and Faculty of Informatics, Technical
  University of Vienna }

\titlerunning{Cut elimination and proof-search for tense logic via
  linear nested sequents}

\begin{document}

\maketitle

\begin{abstract}
  We give a linear nested sequent calculus for the basic normal tense
  logic $\Kt$. We show that the calculus enables backwards
  proof-search, counter-model construction and syntactic
  cut-elimination. Linear nested sequents thus provide the minimal
  amount of nesting necessary to provide an adequate proof-theory for
  modal logics containing converse. As a bonus, this yields a cut-free
  calculus for symmetric modal logic $\KB$.
\end{abstract}

\section{Introduction}

The two main proof-calculi for normal modal logics are sequent calculi
and tableau calculi~\cite{gore-handbook}. Tableau calculi are
algorithmic, directly providing a decision procedure via cut-free
completeness. Sequent calculi are proof-theoretic, requiring us to
show completeness via cut-admissibility. Often, there is a direct
relationship between these two formalisms, where one can be seen as
the ``upside down'' variant of the other. However, this direct
relationship breaks down for modal logics where the modalities are
interpreted with respect to a Kripke reachability relation as well as
its converse relation, as in modal tense logic $\Kt$.

Modal sequent calculi go back to at least
1957~\cite{ohnishi-matsumoto-I}. Sequent calculi for normal modal
tense logics have proved more elusive, with some previous published
attempts failing
cut-elimination~\cite{trzesicki-gentzen-tense-logic-broken}: the
counter-example is $p \to \wbx \lnot \bbx \lnot p$. But we now have
several extended sequent frameworks for tense logics: for example,
display calculi~\cite{DBLP:journals/logcom/Wansing94}; nested
sequents~\cite{DBLP:journals/sLogica/Kashima94,DBLP:journals/corr/abs-1103-5286}
and labelled sequents~\cite{DBLP:conf/ausai/BonnetteG98}.  The main
disadvantage is the rather heavy machinery required to achieve
cut-elimination.  Tableau calculi for tense logics in contrast take a
global view of proof-search, permitting to expand any node in the
search space but requiring technical novelties such as dynamic
blocking~\cite{DBLP:conf/cade/HorrocksST00} and the use of a
``restart'' rule ~\cite{DBLP:conf/tableaux/GoreW09}.

But there is a glaring disparity between the simplicity of tableau
calculi for tense logics versus the mentioned extended sequent
frameworks, giving rise to the question: What is the minimum extension
over traditional sequents enabling a proof-theory for tense logics
amenable to (algorithmic) backward proof-search?

Here, we address this question by giving a sequent-style calculus for
tense logic $\Kt$ which includes two ``restart'' rules.  The calculus
is given in the \emph{linear nested sequent} framework. This
framework, essentially a reformulation of
\emph{2-sequents}~\cite{Masini:1992}, lies between the original
sequent framework and the nested sequent framework, in that it extends
the sequent structure to \emph{lists} of sequents. Apart from
\emph{op.cit.}, this framework yielded, e.g., cut-free calculi for a
number of standard normal and non-normal modal
logics~\cite{Lellmann:2015lns,Lellmann:2019TOCL,Parisi:2017PhD} as
well as temporal or intermediate logics of linear
frames~\cite{Indrzejczak:2016,Kuznets:2018AiML}. Yet, so far the only
examples were logics which either have a cut-free sequent formulation,
or where the underlying semantic structure exactly matches that of
linear nested sequents. The calculus presented here thus is
interesting for two reasons: First, it shows that not the full
complexity of nested sequents is necessary to capture tense logic
without cuts; second, it provides a non-trivial example showing that
the linear nested sequent framework can handle interesting logics
beyond the reach of standard sequents, with models not mirroring the
linear structure.

In the following, we present the calculus, then show how to use it for
backward proof-search and cut-free completeness.  We also show that it
is amenable to the usual proof-theoretic results such as the
admissibility of the structural rules and cut. As a bonus, this yields
a calculus for symmetric modal logic $\KB$, suggesting that the linear
nested sequent framework so far is the simplest purely syntactic
extension of the standard sequent framework capturing $\KB$ in a
cut-free way, since even hypersequent systems for $\KB$, such as that
of Lahav~\cite{DBLP:conf/lics/Lahav13}, seem to require an analytic
cut rule and hence are not completely cut-free.

\section{Preliminaries}

We assume that the reader is familiar with normal modal tense logics
and their associated Kripke semantics but give a very
terse introduction below.

\emph{Formulae} of normal modal tense logics are built from a given set
$\atm$ of atomic formula via the BNF grammar below where $p \in \atm$:
\[
   A := p \mid \bot \mid A \to A \mid \wbx A \mid
   \wdi A \mid \bbx A \mid \bdi A
\]
We assume conjunction, disjunction and negation are defined as usual.

The Kripke semantics for $\Kt$ is given by a non-empty set (of worlds)
$W$, a binary relation $R$ over $W$, and a valuation function $V$
mapping a world $w \in W$ and an atomic formula $p \in \atm$ to either
``true'' or ``false''. Given a Kripke model $\langle W, R, V\rangle$,
the forcing relation $w \forces A$ between a world $w \in W$ and a
formula $A$ is defined as follows (omitting clauses for the
propositional connectives):
\begin{center}
\begin{tabular}[c]{lll@{\hspace{0.5cm}}lll}
  \multicolumn{6}{c}{$w \forces p$ ~ if ~ $V(w, p) = true$}
\\
  $w \forces \wdi A$ & ~if~ & $\exists v \in W.\ w R v ~\&~ v \forces A$
  &   $w \forces \bdi A$ & ~if~ & $\exists v \in W.\ v R w ~\&~ v \forces A$
\\
  $w \forces \wbx A$ & ~if~ & $\forall v \in W.\ w R v ~\Rightarrow~ v \forces A$
  &   $w \forces \bbx A$ & ~if~ & $\forall v \in W.\ v R w ~\Rightarrow~ v \forces A$
\end{tabular}
\end{center}
As usual, a formula $A$ is \emph{satisfiable} if there is some Kripke
model $\langle W, R, V\rangle$, and some world $w \in W$ such that
$w \forces A$. A formula $A$ is \emph{valid} if $\lnot A$ is
unsatisfiable. Formally, the logic $\Kt$ is the set of all valid
formulae.

The traditional Hilbert system $\HKt$ for tense logic $\Kt$ takes all
classical propositional tautologies as axioms, adds the axioms
$\wbx (A \to B) \to (\wbx A \to \wbx B)$ and
$\bbx (A \to B) \to (\bbx A \to \bbx B)$, the necessitation rules
$\mathrm{Nec}_{\wbx}: A / \wbx A$ and
$\mathrm{Nec}_{\bbx}: A / \bbx A$, and the two interaction axioms
$\wdi\bbx p \to p$ and $\bdi\wbx p \to p$. The system $\HKt$ is sound
and complete w.r.t.\ the Kripke semantics.

\section{A Linear Nested Sequent Calculus for $\Kt$}

\begin{figure}[t]
  \hrule\smallskip
  \[
    \infer[\wbx_R^1]{\mathcal{G} \upd \Gamma \sa \Delta \bwd \Sigma
      \sa \Pi, \wbx A
    }
    {\mathcal{G} \upd \Gamma \sa \Delta, A \bwd \Sigma \sa \Pi, \wbx A
      &
      \mathcal{G} \upd \Gamma \sa \Delta \bwd \Sigma\sa \Pi, \wbx A \fwd
      \epsilon\sa A
    }
  \]
  \[
    \infer[\bbx_R^1]{\mathcal{G} \upd \Gamma \sa \Delta \fwd \Sigma
      \sa \Pi, \bbx A
    }
    {\mathcal{G} \upd \Gamma \sa \Delta, A \fwd \Sigma \sa \Pi, \bbx A
      &
      \mathcal{G} \upd \Gamma \sa \Delta \fwd \Sigma\sa \Pi, \bbx A \bwd
      \epsilon\sa A
    }
  \]
  \[
    \infer[\wbx_R^2]{\mathcal{G} \fwd \Gamma \sa \Delta, \wbx A
    }
    {\mathcal{G} \fwd \Gamma \sa \Delta, \wbx A \fwd \epsilon\sa A
    }
    \qquad
    \infer[\bbx_R^2]{\mathcal{G} \bwd \Gamma \sa \Delta, \bbx A
    }
    {\mathcal{G} \bwd \Gamma \sa \Delta, \bbx A \bwd \epsilon\sa A
    }
  \]
  \[
    \infer[\wbx_L^1]{\mathcal{G} \upd \Gamma, \wbx A \sa \Delta \fwd
      \Sigma \sa \Pi
    }
    {\mathcal{G} \upd \Gamma, \wbx A \sa \Delta \fwd
      \Sigma, A \sa \Pi
    }
    \quad
    \infer[\bbx_L^1]{\mathcal{G} \upd \Gamma, \bbx A \sa \Delta \bwd
      \Sigma \sa \Pi
    }
    {\mathcal{G} \upd \Gamma, \bbx A \sa \Delta \bwd
      \Sigma, A \sa \Pi
    }
  \]
  \[
    \infer[\wbx_L^2]{\mathcal{G} \upd \Gamma \sa \Delta \bwd
      \Sigma, \wbx A \sa \Pi
    }
    {\mathcal{G} \upd \Gamma, A \sa \Delta 
    }
    \quad
    \infer[\bbx_L^2]{\mathcal{G} \upd \Gamma \sa \Delta \fwd
      \Sigma, \bbx A \sa \Pi
    }
    {\mathcal{G} \upd \Gamma, A \sa \Delta 
    }
  \]
  \[
    \infer[\idrl]{\mathcal{G} \upd \Gamma, p \sa p, \Delta
    }
    {
    }
    \quad
    \infer[\bot_L]{\mathcal{G} \upd \Gamma, \bot \sa \Delta
    }
    {
    }
    \quad
    \infer[\EW]{\mathcal{G} \upd \Gamma \sa \Delta
    }
    {\mathcal{G}
    }
  \]
  \[
    \infer[\to_R]{\mathcal{G} \upd \Gamma \sa \Delta, A \to B
    }
    {\mathcal{G} \upd \Gamma, A \sa \Delta, A \to B, B
    }
    \qquad
    \infer[\to_L]{\mathcal{G} \upd \Gamma, A \to B \sa \Delta
    }
    {\mathcal{G} \upd \Gamma, A \to B, B \sa \Delta
      &
      \mathcal{G} \upd \Gamma, A \to B \sa \Delta, A
    }
  \]
  \hrule
  \caption{The system $\LNS_\Kt$ where $\upd$ stands for either $\fwd$ or $\bwd$}
  \label{fig:rules-kb}
\end{figure}

Unlike standard Hilbert-calculi, our calculus operates on linear
nested sequents instead of formulae, defined and adapted from
Lellmann~\cite{Lellmann:2015lns} as follows.

\begin{definition}
  A \emph{component} is an expression $\Gamma \stt \Delta$ where the
  \emph{antecedent} $\Gamma$ and the \emph{succedent} $\Delta$ are
  finite, possibly empty, multisets of formulae.  We write $\epsilon$
  to stand for an empty antecedent or succedent to avoid confusion. A
  \emph{linear nested sequent} is an expression obtained via the
  following BNF grammar:
  \[
    S
    := \seq{\Gamma}{\Delta} \mid  \seq{\Gamma}{\Delta} \fwd  S \mid
    \seq{\Gamma}{\Delta} \bwd S \;.
  \]
\end{definition}

We often write $\mathcal{G}$ for a possibly empty \emph{context}:
e.g., $\mathcal{G} \fwd \Gamma \stt \Delta$ stands for
$\Gamma \stt \Delta$ if $\mathcal{G}$ is empty, and for
$\Sigma \stt \Pi \bwd \Omega \stt \Theta\fwd \Gamma \stt \Delta$ if
$\mathcal{G}$ is the linear nested sequent
$\Sigma \stt \Pi \bwd \Omega \stt \Theta$.  Fig.~\ref{fig:rules-kb}
shows the rules of our calculus $\LNS_\Kt$. As usual, each rule has a
number of \emph{premisses} above the horizontal line and a single
\emph{conclusion} below it.  The single formula in the conclusion is
the \emph{principal} formula and the formulae in the premisses are the
\emph{side-formulae}.

Every instance of the rule (id) is a \emph{derivation} of height 0,
and if $(\rho)$ is an $n$-ary rule and we are given $n$ premiss
derivations $d_1, \cdots , d_n$, each of height $h_1, \cdots, h_n$,
with respective conclusions $c_1, \cdots, c_n$, and
$c_1, \cdots , c_n/d_0$ is an instance of $(\rho)$ then
$d_1, \cdots , d_n / d_0$ is a derivation of height
$1 + max\{h_1, \cdots , h_n\}$.  We write $\mathcal{D} \entails S$ if
$\mathcal{D}$ is a derivation in $\LNS_\Kt$ of the linear nested
sequent $S$, and $\entails S$ if there is a derivation $\mathcal{D}$
with $\mathcal{D} \entails S$.

Note that our calculus is \emph{end-active}, i.e., in every logical
rule and every premiss, at least one active formula occurs in the last
component.

\begin{example}
  Consider the end-sequent
  $\seq{}{\wbx p, \wbx q, r \to\bbx\lnot\wbx\lnot r}$ where
  $r \to \bbx\lnot\wbx\lnot r$ is the axiom $r \to \bbx\wdi r$ with
  the definition of $\wdi$ expanded.  Suppose we apply the rule
  $(\to_R)$ upward to obtain
  $\seq{r}{\wbx p, \wbx q, \bbx \lnot\wbx\lnot r}$.  Then there are
  two different instances of the rule $\wbx_R^2$ using two different
  principal formulae, neither of which leads to a derivation,
  and one instance of the rule $\bbx_R^2$ which leads to a derivation:\\
\begin{minipage}[c]{0.5\linewidth}
\[
\AxiomC{
$\seq{r}{\wbx p, \wbx q, \bbx \lnot\wbx\lnot r} \fwd \seq{\epsilon}{p}$
}
\RightLabel{$\wbx_R^2$}
\UnaryInfC{$\seq{r}{\wbx p, \wbx q, \bbx \lnot\wbx\lnot r}$}
\DisplayProof
\]
\[
\AxiomC{
$\seq{r}{\wbx p, \wbx q, \bbx \lnot\wbx\lnot r} \fwd \seq{\epsilon}{q}$
}
\RightLabel{$\wbx_R^2$}
\UnaryInfC{$\seq{r}{\wbx p, \wbx q, \bbx \lnot\wbx\lnot r}$}
\DisplayProof
\]
\end{minipage}
\begin{minipage}[c]{0.5\linewidth}
\[
\AxiomC{~}
\RightLabel{id}
\UnaryInfC{
$\seq{r, \lnot r}{r, \wbx p, \wbx q, \bbx \lnot\wbx\lnot r}$
}
\RightLabel{$\lnot_L$}
\UnaryInfC{
$\seq{r, \lnot r}{\wbx p, \wbx q, \bbx \lnot\wbx\lnot r}$
}
\RightLabel{$\wbx_L^2$}
\UnaryInfC{
$\seq{r}{\wbx p, \wbx q, \bbx \lnot\wbx\lnot r} \bwd \seq{\wbx \lnot
  r}{\lnot \wbx \lnot r}$
}
\RightLabel{$\lnot_R$}
\UnaryInfC{
$\seq{r}{\wbx p, \wbx q, \bbx \lnot\wbx\lnot r} \bwd \seq{\epsilon}{\lnot\wbx\lnot r}$
}
\RightLabel{$\bbx_R^2$}
\UnaryInfC{$\seq{r}{\wbx p, \wbx q, \bbx \lnot\wbx\lnot r}$}
\DisplayProof
\]
\end{minipage}
\end{example}

Intuitively, each component of a linear nested sequent corresponds to
a world of a Kripke model, and the structural connectives $\fwd$ and
$\bwd$ between components corresponds to the relations $R$ and
$R^{-1}$ that connect these worlds.

These intuitions can be made formal since linear nested sequents have
a natural interpretation as formulae given by taking $\fwd$ and $\bwd$
to be the structural connectives corresponding to $\wbx$ and $\bbx$,
respectively:

\begin{definition}
  If $\Gamma = \{A_1, \cdots , A_n\}$ then we write $\hat{\Gamma}$ for
  $A_1 \land \cdots \land A_n$ and $\breve{\Gamma}$ for
  $A_1 \lor \cdots \lor A_n$. The \emph{formula translation} of a
  linear nested sequent is given recursively by
  $\tau(\Gamma \stt \Delta) = \hat{\Gamma} \to \breve{\Delta}$ and
  \[
    \tau(\Gamma \stt \Delta \fwd \mathcal{G}) = \hat{\Gamma} \to
    (\breve{\Delta} \lor \wbx\; \tau(\mathcal{G})) \quad
    \tau(\Gamma \stt \Delta \bwd \mathcal{G}) =
    \hat{\Gamma} \to (\breve{\Delta} \lor \bbx\;
    \tau(\mathcal{G}))\;.
  \]
  A sequent $S$ is \emph{falsifiable} if there exists a model
  $\langle W, R, v \rangle$ and a world $w \in W$ such that
  $w \not\forces \tau(S)$.  A sequent $S$ is \emph{valid} if it is not
  falsifiable.
\end{definition}

Soundness of the calculus then follows by induction on the depth of
the derivation from the following theorem.

\begin{theorem}[Soundness]\label{thm:soundness}
  For every rule, if the conclusion is falsifiable then so is one of
  the premisses.
\end{theorem}

\begin{proof}
  We only give the interesting cases going beyond the standard
  calculi.

  For rule $\wbx_R^1$, suppose that for
  $\mathfrak{M} = \langle W,R,V \rangle$ and $w_1 \in W$ we have
  $\mathfrak{M},w_1 \nforces \transl[\Gamma_1 \sa \Delta_1 \upd \dots
  \upd \Gamma_n \sa \Delta_n \upd \Gamma \sa \Delta \bwd \Sigma \sa
  \Pi, \wbx A]$. Hence there are worlds $w_2,\dots,w_n,x,y \in W$ with
  $w_1 R^{\epsilon_1} w_2 R^{\epsilon_2} \dots R^{\epsilon_{n-1}} w_n
  R^{\epsilon_n} x R^{-1} y$, for $\epsilon_i$ empty or $-1$ as
  needed, such that
  $w_i \forces \hat{\Gamma_i} \land \neg \breve{\Delta_i}$ for every
  $i \leq n$, as well as
  $x \forces \hat{\Gamma} \land \neg \breve{\Delta}$ and
  $y \forces \hat{\Sigma} \land \neg \breve{\Pi} \land \neg \wbx
  A$. Hence there is a world $z \in W$ with $yRz$ such that
  $z \nforces A$. If $z = x$, then $\mathfrak{M},w_1$ falsifies the
  interpretation of the first premiss.  If $z \neq x$, we have a model
  falsifying the interpretation of the second premiss. The case of
  rule $\bbx_R^1$ is analogous.
  
  For the ``restart'' rule $(\bbx_L^2)$ , suppose that the conclusion
  $\mcalG \upd \seq{\Gamma}{\Delta} \fwd \seq{\Sigma, \bbx A}{\Pi}$ is
  falsifiable. Thus there is a world $w$ such that
  $w \not\forces \hat{\Gamma} \to \breve{\Delta} \lor \wbx
  (\hat{\Sigma} \land \bbx A \to \breve{\Pi})$.  So
  $w \forces \hat{\Gamma}$ and $w \not\forces \breve{\Delta}$ and $w$
  must have an $R$-successor $v$ such that $v \forces \hat{\Sigma}$
  and $v \not\forces \breve{\Pi}$ and $v \forces \bbx A$.  But then
  $w \forces A$, exactly as desired to conclude that the premiss
  $\seq{\Gamma, A}{\Delta}$ is falsifiable.\qed
\end{proof}

\begin{corollary}
  For every linear nested sequent $S$, if $\entails S$, then $\tau(S)$
  is valid.\qed
\end{corollary}

Why does the premiss of the rule $\bbx_L^2$ not contain the sequent
$\Sigma, \bbx A \stt \Pi$~? Because there may be an incompatibility
between $w$ and its $R$-successor $v$. The $\bbx_L^2$ rule removes
this incompatibility by propagating $A$ to the $R$-predecessor
$w$. But $A$ could be arbitrarily complex and we must again saturate
the predecessor before re-creating $v$. The current $v$ must be
deleted and we must ``restart'' $w$.

Before showing completeness of $\LNS_\Kt$ we remark on a
simplification of the calculus. Let $\LNS_\Kt^*$ be the calculus
obtained from $\LNS_\Kt$ by replacing the modal right rules
$\wbx_R^1, \bbx_R^1,\wbx_R^2$ and $\bbx_R^2$ with the following two
rules:
\[
  \infer[\wbx_R]{\mathcal{G} \upd \Gamma \sa \Delta, \wbx A
  }
  {\mathcal{G} \upd \Gamma \sa \Delta, \wbx A \fwd \epsilon\sa A
  }
  \qquad
  \infer[\bbx_R]{\mathcal{G} \upd \Gamma \sa \Delta, \bbx A
  }
  {\mathcal{G} \upd \Gamma \sa \Delta, \bbx A \bwd \epsilon\sa A
  }
\]
Soundness of these rules can be shown exactly as in
Thm.~\ref{thm:soundness}. Moreover, since derivations in the system
$\LNS_\Kt$ can be converted straightforwardly into derivations in the
system $\LNS_\Kt^*$ by simply omitting the subderivations of the left
premisses of $\wbx_R^1$ and $\bbx_R^1$ respectively, we immediately
obtain:

\begin{proposition}
  If $\LNS_\Kt$ is cut-free complete for $\Kt$, then so is
  $\LNS_\Kt^*$.\qed
\end{proposition}

For technical reasons, in particular to facilitate a cut elimination
proof when the cut formula is principal in the rules $\wbx_L^2$ or
$\bbx_L^2$, in the following we take $\LNS_\Kt$ as the main system,
but it is worth keeping in mind that the completeness results
automatically extend to $\LNS_\Kt^*$. Note also that, modulo the
structural rules and deleting the last component in the rules
$\wbx_L^2$ and $\bbx_L^2$, $\LNS_\Kt^*$ is essentially a two-sided
linear end-active reformulation of the cut-free nested sequent
calculus $S2K_t$ for $\Kt$
in~\cite{DBLP:journals/sLogica/Kashima94}. Hence completeness of the
latter follows from our completeness results by transforming
derivations bottom-up.

\section{Completeness via proof search and counter-models}

We now show how to use our calculus (without $\EW$) for backward proof
search, and how to obtain a counter-model from failed proof search,
yielding completeness. For this, we separate the rules into groups,
assuming an appropriate side-condition to ensure that rules are
applied only when they create new formulae:
\begin{description}
\item[\em Termination Rules:] \idrl{} and $\bot_L$;
\item[\em CPL Rules:] $(\to_R)$ and $(\to_L)$. The side-conditions
  ensuring termination are: $A \not\in \Gamma$ or $B \not\in \Delta$
  for $(\to_R)$, and $B \not \in \Gamma$ and $A \not \in \Delta$ for
  $(\to_L)$;
\item[\em Propagation Rules:] $\wbx_L^1$ and $\bbx_L^1$. These rules
  move subformulae to the last component.  The side-condition ensuring
  termination is that $A \not \in \Sigma$;
\item[\em Restart rules:] $\wbx_L^2$ and $\bbx_L^2$. These rules make
  the sequent shorter. The side-condition ensuring termination is that
  $A \not\in\Gamma$;
\item[\em Box Rules:] $\wbx_R^1$, $\wbx_R^2$, $\bbx_R^1$
  $\bbx_R^2$. We apply only one of these rules, even if many are
  applicable, and backtrack over these choices. But these rules are
  non-deterministic since they choose a particular formula as
  principal. We must also back-track over all choices of principal
  formula in the chosen rule.
\end{description}
Our \emph{proof-search strategy} is to apply (backwards) the highest
rule in the above list. Thus, assuming that the \idrl{} rule is not
applicable, our strategy first seeks to saturate the final component
with the CPL-rules. Then we seek to propagate formulae from the
second-final component into the final component. Then we seek to
repair any incompatibilities between the final two components using
the Restart rules to shorten the sequent if necessary. Only when none
of these rules are applicable do we apply a Box-rule to lengthen the
sequent, and backtrack over all choices of principal formula.  In
particular, if a node is ``restarted'' then we have to redo all
previous Box-rule applications from this changed node.

Overall, the strategy means that the maximal modal degree, defined
standardly, of a formula in a component must decrease strictly as the
sequent becomes longer, and the restart rules, which shorten the
sequent, do not increase this maximal modal degree. A particular
component is restarted only a finite number of times because each
restart adds a formula which is a strict subformula of the
end-sequent, and there are only a finite number of these. Hence the
proof-search terminates.

\begin{theorem}[Termination]\label{thm:termination}
  Backward proof-search terminates.\qed
\end{theorem}

Suppose backward proof-search terminates without finding a derivation.
How do we construct a counter-model that falsifies the end-sequent?
Consider the search-space explored by our procedure, i.e., the space
of all possible failed derivations including the various backtracking
choice-points inherent in the search procedure. We visualise this
search space as a single tree by conjoining the modal rules containing
backtrack choices. E.g., the backtracking choices in the sequent
$\seq{\epsilon}{\wbx p, \wbx q, \wbx r}$ can be ``determinised'' as
below where we have used ``dotted'' lines to indicate a meta-level
conjunction which ``binds'' the three premisses:
\[
\begin{adjustbox}{max width=0.98\textwidth}
\AxiomC{
$\seq{\epsilon}{\wbx p, \wbx q, \bbx r} \fwd \seq{\epsilon}{p}$
\quad
$\seq{\epsilon}{\wbx p, \wbx q, \bbx r} \fwd \seq{\epsilon}{q}$
\quad
$\seq{\epsilon}{\wbx p, \wbx q, \bbx r} \bwd \seq{\epsilon}{r}$
}
\dottedLine
\UnaryInfC{$\seq{\epsilon}{\wbx p, \wbx q, \bbx r}$}
\DisplayProof
\end{adjustbox}
\]
Similarly, the sequent
$\mcalG \upd \seq{\Gamma_1}{\Delta_1 \bwd} \seq{\Gamma_2}{\wbx p, \bbx
  q}$ can be determinised as:
\[
\AxiomC{(a)}
\AxiomC{$\mcalG \upd \seq{\Gamma_1}{\Delta_1} \bwd \seq{\Gamma_2}{\wbx p,\bbx q}
  \bwd \seq{\epsilon}{q}$}
\dottedLine
\BinaryInfC{$\mcalG \upd \seq{\Gamma_1}{\Delta_1} \bwd
  \seq{\Gamma_2}{\wbx p, \bbx q}$}
\DisplayProof
\]
with (a) being the pair below:
\[
\begin{adjustbox}{max width=0.98\textwidth}
\AxiomC{$\mcalG \upd \seq{\Gamma_1}{p,\Delta_1} \bwd \seq{\Gamma_2}{\wbx p, \bbx q}$}
\AxiomC{$\mcalG \upd \seq{\Gamma_1}{\Delta_1 \bwd} \seq{\Gamma_2}{\wbx p, \bbx q} \fwd \seq{\epsilon}{p}$}
\LeftLabel{$\wbx_R^1$}
\RightLabel{$p \not\in \Delta_1$}
\BinaryInfC{(a)}
\DisplayProof
\end{adjustbox}
\]
We dub these choice-points as \emph{``and-nodes''} to distinguish them
from the traditional \emph{``or-nodes''} created by
disjunctions~\cite{DBLP:conf/tableaux/GoreW09}.
We first show how we prune this search space to keep only nodes useful
for building a counter-model. We then outline how the pruned search
space yields a counter-model for the end-sequent.

\subsection{Pruning irrelevant branches from the search space}

Suppose the original search-space corresponds to a tree $\tau_0$, and
consider some leaf to which no rule is applicable.  In this search
tree, delete all the rightmost components of the conclusion of a
restart rule.  We can do so because we know that, in the conclusion,
the second-last component is incompatible with the last component
precisely because its antecedent $\Gamma$ is missing $A$. So this pair
of components cannot possibly be part of a counter-model.

Now consider the rule application $(\rho)$ below the restart
rule. Suppose the last component of the premiss of $(\rho)$ is
$\seq{\Sigma, \wbx A}{\Pi}$.  If deleting $\seq{\Sigma,\wbx A}{\Pi}$
causes $(\rho)$ to become meaningless, then delete the last component
of the conclusion of $(\rho)$.  If the rule is binary or is an
``and-rule'' then we keep the shorter of the sequents that are
returned downward by this procedure.  E.g., an instance of the rule
$\wbx_L^2$ from Fig.~\ref{fig:rules-kb}, as below, now appears as
shown below it:
\[
\AxiomC{$\mcalG \upd \seq{A, \Gamma}{\Delta}$}
\RightLabel{$A \not\in \Gamma$}
\LeftLabel{$\wbx_L^2$}
\UnaryInfC{$\mcalG \upd \seq{\Gamma}{\Delta} \fwd \seq{\wbx A, \Sigma}{A}$}
\AxiomC{$\vdots$}
\UnaryInfC{$\mcalG \upd \seq{\Gamma}{\Delta} \fwd \seq{\wbx A, \Sigma}{B}$}
\LeftLabel{$\land_R$}
\BinaryInfC{$\mcalG \upd \seq{\Gamma}{\Delta} \fwd \seq{\wbx A, \Sigma}{A \land B}$}
\DisplayProof
\]
\[
\AxiomC{$\mcalG \upd \seq{A, \Gamma}{\Delta}$}
\UnaryInfC{$\mcalG \upd \seq{\Gamma}{\Delta}$}
\AxiomC{$\mcalG'$}
\BinaryInfC{$\mcalG''$}
\DisplayProof
\]
where $\mcalG'$ is the pruned version of
$\mcalG \upd \seq{\Gamma}{\Delta} \fwd \seq{\wbx A, \Sigma}{B}$ and
$\mcalG''$ is the shorter of $\mcalG \upd \seq{\Gamma}{\Delta}$ and
$\mcalG'$.  We can do so because the shorter sequent $\mcalG''$
restarts a component that is earlier in the order of expansion, hence
closer to the initial sequent.  Now proceed by considering the number
of restarts.

\begin{lemma}
  For all $\Gamma$ and $\Delta$, if $\seq{\Gamma}{\Delta}$ is not
  derivable and no restart rule is ever applied then there exists a
  Kripke model which falsifies $\Gamma \stt \Delta$.
\end{lemma}

\begin{proof}
  If no restart rules are applied in backward proof-search, then every
  application of a Box-right-rule leads to a new component which is
  compatible with its parent component in that every required formula
  is already in the latter.

  Now consider any three adjacent components of a leaf sequent, which must
  be of one of the following forms where the second-last component and the
  third-last component are separated by $\fwd$ (we skip the 
  similar cases when it is
  $\bwd$):

  \begin{description}
  \item[\rm (1)] 
    $\mcalG \upd \seq{\Gamma_0, \Sigma_1, \wbx\Sigma_5}{D_l,\Delta_0}\fwd$
    \\
    $~~~~~~~~~~~~~~~~~~~~~~~~\seq{\Gamma_1, \bbx \Sigma_1, \Sigma_2, \wbx\Sigma_3,\Sigma_5,\Sigma_4}{\Delta_1, \bbx A_i, \wbx B_j, \wbx C_k, \bbx D_l}$
    \\
    $~~~~~~~~~~~~~~~~~~~~~~~~\bwd \seq{\Gamma_2, \Sigma_1, \wbx
      \Sigma_2}{\Delta_2, A_i}  \upd \mcalH$
  \item[\rm (2)] 
    $\mcalG \upd \cdots\cdots\cdots\cdots\cdots\cdots\cdots\fwd$
    \\
    $~~~~~~~~~~~~~~~~~~~~~~~~\cdots\cdots\cdots\cdots\cdots\cdots\cdots
    \cdots\cdots\cdots\cdots\cdots\cdots\cdots\cdots\cdots\cdots$
    \\
    $~~~~~~~~~~~~~~~~~~~~~~~\fwd
    \seq{\Gamma_2',\Sigma_3,\bbx\Sigma_4}{\Delta_2', C_k}  \upd
    \mcalH$
  \item[\rm (3)] 
    $\mcalG \upd \cdots\cdots\cdots\cdots\cdots\cdots\cdots\fwd$
    \\
    $~~~~~~~~~~~~~~~~~~~~~~~~\cdots\cdots\cdots\cdots\cdots\cdots\cdots
    \cdots\cdots\cdots\cdots\cdots\cdots\cdots\cdots\cdots\cdots$
    \\
    $~~~~~~~~~~~~~~~~~~~~~~~\fwd
    \seq{\Gamma_2'',\Sigma_3,\bbx\Sigma_4}{\Delta_2'', B_j}  \upd \mcalH$
  \end{description}
  In (1), the final component is the right premiss of the rule
  $\bbx_R^1$ on $\bbx A_i$, so $\bbx A_i$ is ``fulfilled''.  The rule
  $\wbx_L^2$ is not applicable to the last component because
  $\Sigma_2$ is in the middle component.  The rule $\bbx_L^1$ is not
  applicable on the middle component because $\Sigma_1$ is in the last
  component. The rule $\bbx_L^2$ is not applicable to the middle
  component because $\Sigma_1$ is also in the first component.  The
  rule $\wbx_L^1$ is not applicable on the first component because
  $\Sigma_5$ is in the middle component.  The $\bbx D_l$ in the middle
  component is fulfilled because the first component contains $D_l$
  via the left premiss of $\bbx_R^1$.

  The two formulae $\wbx B_j$ and $\wbx C_k$ in the middle component
  are not fulfilled by (1). But there will be an application of
  $\wbx_R^2$ on $\wbx C_k$ shown as (2), and another similar instance
  on $\wbx B_j$ with $C_k$ in the last component replaced by $B_j$.
  Rule $\bbx_L^2$ is not applicable on the last component because
  $\Sigma_4$ is in the middle one.  The rule $\wbx_L^1$ is not
  applicable on the middle component because $\Sigma_3$ is in the last
  one.

  These arguments apply for every $\wbx$-formula and for every
  $\bbx$-formula in the second-last component.  Moreover, for every
  conjunction in the succedent of either component, at least one
  conjunct must be in that succedent. Similarly, for every disjunction
  in the succedent of either component, both disjuncts must be in that
  succedent.  Finally, the \idrl{} rule is not applicable to any
  component.

  Now put the following valuation on these components: every formula
  in the antecedent has a value of ``true'' and every formula in the
  succedent has a value of ``false''. Then replace every occurrence of
  $\fwd$ with $R$ and replace every occurrence of $\bwd$ with
  $R^{-1}$.  Thus we have the following picture:
  \[
    \xymatrix@=0.4cm{
      w_j: \seq{\Gamma_2'', \Sigma_3, \bbx\Sigma_4}{\Delta_2'', B_j}
      & 
      w_k: \seq{\Gamma_2',\Sigma_3,\bbx\Sigma_4}{\Delta_2', C_k}
      \\
      v: \Gamma_1, \bbx \Sigma_1, \Sigma_2, \wbx\Sigma_3, \Sigma_5,\Sigma_4
      \qquad\seq{}{}
      \ar[u]^{R}
      & {\Delta_1, \bbx A_i,\wbx B_j, \wbx C_k,\bbx D_l}
      \ar[u]_{R}
      \\
      u: \seq{\Gamma_0, \Sigma_1, \wbx\Sigma_5}{D_l,\Delta_0}
      \ar[u]_{R}
      & 
      w_i: \seq{\Gamma_2, \Sigma_1, \wbx\Sigma_2}{\Delta_2, A_i}
      \ar[u]^{R}
    }
  \]

  For every world $v$, every formula $\bbx A_i$ and every formula
  $\wbx C_k$ with $v \not\forces \bbx A_i$ and
  $v \not\forces \wbx C_k$, there exists a predecessor world $u_i$
  with $u_i R v$ and $u_i \not\forces A_i$, there exists a successor
  world $w_k$ with $v R w_k$ and $w_k \not\forces C_k$; for every
  formula $\wbx B \in \wbx\Sigma_2$ with $u_i \forces \wbx B$, we have
  $v \forces B$; and for every formula $\bbx D \in \bbx\Sigma_4$ with
  $w_k \forces \bbx D$, we have $v \forces D$.  Hence, the triple
  $u_i R v R w_k$ is mutually compatible in terms of both modalities,
  and each world falsifies the associated component. Similar triples
  exist for all the box-formulae in $v$ which are not principal in the
  diagram, and they all ``overlap'' at $v$. Hence we can ``glue'' them
  together to form the fan of R-successors and R-predecessors of $v$,
  maintaining global compatibility.  The original sequent
  $\seq{\Gamma}{\Delta}$ is thus falsified at its associated
  world. \qed
\end{proof}

\begin{lemma}\label{lem:countermodel}
  For every $\Gamma$ and $\Delta$, if the sequent
  $\seq{\Gamma}{\Delta}$ is not derivable, and contains restarts, then
  there is a Kripke model which falsifies the end-sequent.
\end{lemma}

\begin{proof}
  We proceed by induction on the number of restarts. If there are
  none, then we are done by the previous lemma. Else there are a
  finite number of restarts.
  
  Consider the highest restart and suppose it is $\wbx_L^2$. By our
  deletion strategy, it must look exactly as shown above. By the
  induction hypothesis, the premiss must have a counter-model. But the
  premiss is a strict superset of the conclusion, so the same model
  must falsify the conclusion. \qed
\end{proof}

\begin{example}
  Consider the end-sequent $\seq{\epsilon}{\wbx p, \wbx q, \bbx
    r}$. We would need two successor worlds, falsifying $p$ and $q$
  respectively, and one predecessor world falsifying $r$. One failed
  derivation will come from
  $\seq{\epsilon}{\wbx p, \wbx q, \wbx r} \fwd \seq{\;}{p}$ while
  another will come from
  $\seq{\epsilon}{\wbx p, \wbx q, \bbx r} \fwd \seq{\;}{q}$, i.e., two
  instances of the $\wbx_R^2$-rule. But there will also be a failed
  derivation from
  $\seq{\epsilon}{\wbx p, \wbx q, \bbx r} \bwd \seq{\;}{r}$, i.e., an
  instance of the $\bbx_R^2$-rule. Moreover, if $r := \lnot\wbx r'$
  then the failed derivation of this last mentioned sequent will have
  a backward application of $\wbx_L^2$ above it, containing a failed
  derivation for $\seq{r'}{\wbx p, \wbx q, \bbx \lnot\wbx r'}$,
  thereby ensuring compatibility.  But there will also be failed
  derivations for $\seq{r'}{\wbx p, \wbx q, \wbx r} \fwd \seq{\;}{p}$
  and $\seq{r'}{\wbx p, \wbx q, \bbx r} \fwd \seq{\;}{q}$ and the
  witnesses for $\wbx p$ and $\wbx q$ will come from these failed
  derivations, because the returned sequent
  $\seq{r'}{\wbx p, \wbx q, \bbx \lnot\wbx r'}$ will be shorter than
  the other ``and-node'' premisses
  $\seq{\epsilon}{\wbx p, \wbx q, \wbx r} \fwd \seq{\;}{p}$ and
  $\seq{\epsilon}{\wbx p, \wbx q, \bbx r} \fwd \seq{\;}{q}$.  But note
  that a counter-model for
  $\seq{r'}{\wbx p, \wbx q, \bbx \lnot\wbx r'}$ is also a
  counter-model for the end-sequent
  $\seq{\epsilon}{\wbx p, \wbx q, \bbx \lnot\wbx r'}$.
\end{example}

Putting Thm.~\ref{thm:termination} and Lem.~\ref{lem:countermodel}
together we then obtain cut-free completeness:

\begin{theorem}[Cut-free Completeness]
  If backward proof-search on end-sequent $S$ fails to find a
  derivation then there is a counter-model for $S$.\qed
\end{theorem}

\begin{corollary}
  If $\varphi$ is valid then the end-sequent $\seq{\epsilon}{\varphi}$
  is derivable.\qed
\end{corollary}

It is tempting to think that we need some sort of coherence condition
as illustrated by the tree in Fig.~\ref{fig:coherence-example1}:
    
\begin{figure}[t]
  \hrule
  \smallskip
  \[
    \infer=[]{\epsilon \stt\underbrace{\wbx(\underbrace{(p \land \neg p) \lor \neg \bbx\wbx q}_{\varphi}), \wbx(\underbrace{(r \land \neg r)
    \lor \neg \bbx \wbx s}_{\psi})}_{\Gamma}
    }
    {\infer[\land_R]{\;\stt \Gamma \fwd \bbx \wbx q \stt[\phantom{\wbx \varphi}] p
        \land \neg p
      }
      {\infer[]{\;\stt \Gamma \fwd \bbx \wbx q \stt[\phantom{\wbx\varphi}] p
        }
        {\infer=[]{\wbx q \stt \Gamma
          }
          {\infer[]{\wbx q \stt \Gamma \fwd \bbx \wbx s \stt[\phantom{\wbx \psi}] r \land \neg r
            }
            {\infer[\land_R]{\wbx q \stt \Gamma \fwd \bbx \wbx s, q \stt[\phantom{\wbx \psi}] r \land \neg r
              }
              {\dots
                &
                \infer[]{\wbx q \stt \Gamma \fwd \bbx \wbx s, q \stt[\phantom{\wbx\psi}] \neg r
                }
                {\infer=[]{\wbx q, \wbx s \stt \Gamma
                  }
                  {\infer[\land_R]{\wbx q, \wbx s \stt \Gamma \fwd \bbx \wbx
                      q, q, s \stt[\phantom{\wbx \varphi}] p \land \neg p
                    }
                    {\dots
                      &
                      \infer*[]{\wbx q, \wbx s \stt \Gamma \fwd \bbx \wbx
                      q, q, s \stt[\phantom{\wbx \varphi}] \neg p
                      }
                      {
                      }
                    }
                  }
                }
              }
            }
          }
        }
        &
        \dots
      }
    }
  \]
  \hrule
  \caption{}\label{fig:coherence-example1}
\end{figure}
In the lowermost application of $\land_R$ we choose the left premiss,
and in the uppermost one the right one. Thus it seems that in the
world corresponding to these last components we would need to make
both $p$ and $\neg p$ true, which of course would not work. But our
pruning turns this failed derivation tree into the tree in
Fig.~\ref{fig:coherence-example2}.
  \begin{figure}[t]
    \hrule
    \smallskip
  \[
    \infer=[]{\epsilon \stt\underbrace{\wbx(\underbrace{(p \land \neg p) \lor \neg \bbx\wbx q}_{\varphi}), \wbx(\underbrace{(r \land \neg r)
    \lor \neg \bbx \wbx s}_{\psi})}_{\Gamma}
    }
    {\infer[]{\;\stt \Gamma 
      }
      {\infer[]{\;\stt \Gamma
        }
        {\infer=[]{\wbx q \stt \Gamma
          }
          {\infer[]{\wbx q \stt \Gamma
            }
            {\infer[]{\wbx q \stt \Gamma
              }
              {\dots
                &
                \infer[]{\wbx q \stt \Gamma
                }
                {\infer=[]{\wbx q, \wbx s \stt \Gamma
                  }
                  {\infer[\land_R]{\wbx q, \wbx s \stt \Gamma \fwd \bbx \wbx
                      q, q, s \stt[\phantom{\wbx \varphi}] p \land \neg p
                    }
                    {\dots
                      &
                      \infer*[]{\wbx q, \wbx s \stt \Gamma \fwd \bbx \wbx
                      q, q, s \stt[\phantom{\wbx \varphi}] \neg p
                      }
                      {
                      }
                    }
                  }
                }
              }
            }
          }
        }
        &
        \dots
      }
    }
  \]
  \hrule
  \caption{}\label{fig:coherence-example2}
\end{figure}
Note that only the component which is not restarted survives the
pruning.  The previous incarnation of the component caused the
restart, but the restarted node did not necessarily follow the same
sequence of rule applications, once it was restarted. Indeed, the
sequence may no longer be possible as it may lead to an instance of
\idrl. Of course, if it is possible and remains open, then it will
find a counter-model for a larger set, which will also suffice for the
smaller set.  Thus our backward proof-search procedure creates
surviving successors/predecessors only when it has ensured that they
will be compatible via some number of restarts. Their incarnations
which are not compatible are irrelevant, and are deleted by our
counter-model construction.

\section{Completeness via Cut elimination}
\label{sec:cut-elimination}

We now provide an alternative proof of cut-free completeness of our
calculus via syntactic cut elimination. The proof is interesting from
a technical point of view: The additional left premiss in the rules
$\wbx_R^1$ and $\bbx_R^1$ is introduced specifically as a counterpart
to the restart rules $\wbx_L^2$ and $\bbx_L^2$ to facilitate the
reduction of cuts on boxed formulae to cuts of smaller complexity.
However, while this enables the cut elimination proof itself, it
shifts a large part of the work in the completeness proof to a perhaps
unexpected place: the proof of admissibility of necessitation.

The following two lemmata are shown straightforwardly by induction on
the depth of the derivation and the complexity of the formula $A$,
respectively:

\begin{lemma}
  The rules  below are admissible in $\LNS_\Kt$:
  \[
    \infer[\W]{\mathcal{G} \upd \Gamma, \Sigma \sa \Delta, \Pi \upd \mathcal{H}
    }
    {\mathcal{G} \upd \Gamma \sa \Delta \upd \mathcal{H}
    }
    \quad
    \infer[\ConL]{\mathcal{G} \upd \Gamma, A \sa \Delta \upd \mathcal{H}
    }
    {\mathcal{G} \upd \Gamma, A, A \sa \Delta \upd \mathcal{H}
    }
    \quad
    \infer[\ConR]{\mathcal{G} \upd \Gamma \sa \Delta, A \upd \mathcal{H}
    }
    {\mathcal{G} \upd \Gamma \sa \Delta, A, A \upd \mathcal{H}
    }
  \]
\qed
\end{lemma}

\begin{lemma}\label{lem:generalised-init}
  The \emph{generalised initial sequent rule} shown below is derivable
  in $\LNS_\Kt$:
  \[
    \infer[]{\mathcal{G} \upd \Gamma, A \sa A, \Delta}{}
  \]
  \qed
\end{lemma}

In order to introduce cuts in our framework, we need the following
notion.

\begin{definition}
  The \emph{merge} of two linear nested sequents is defined via the
  following, where we assume $\mathcal{G},\mathcal{H}$ to be nonempty:
  \begin{align*}
    (\Gamma \sa \Delta) \mrg (\Sigma \sa \Pi ) &:=  \Gamma, \Sigma
                                                 \sa \Delta, \Pi\\
    (\Gamma \sa \Delta) \mrg (\Sigma \sa \Pi \upd \mathcal{H}) &:=  \Gamma, \Sigma
                                                                 \sa \Delta, \Pi
                                                                 \upd \mathcal{H}\\
    (\Gamma \sa \Delta \upd \mathcal{H})  \mrg (\Sigma \sa \Pi) &:=  \Gamma, \Sigma
                                                                  \sa \Delta, \Pi
                                                                  \upd \mathcal{H}\\
    (\Gamma \sa \Delta \fwd \mathcal{G}) \mrg (\Sigma \sa \Pi \fwd
    \mathcal{H}) &:=  \Gamma, \Sigma \sa \Delta, \Pi \fwd (\mathcal{G}
                   \mrg \mathcal{H})\\
    (\Gamma \sa \Delta \bwd \mathcal{G}) \mrg (\Sigma \sa \Pi \bwd
    \mathcal{H}) &:=  \Gamma, \Sigma \sa \Delta, \Pi \bwd (\mathcal{G}
                   \mrg \mathcal{H})\;.
  \end{align*}
\end{definition}
Hence the merge is only defined for linear nested sequents which are
\emph{structurally equivalent}, i.e., have the same structure of the
nesting operators.

Recall that we write $\mathcal{D} \entails \mathcal{G}$ if
$\mathcal{D}$ is a derivation in $\LNS_\Kt$ of the linear nested
sequent $\mathcal{G}$, and $\entails \mathcal{G}$ if there is a
derivation $\mathcal{D}$ with $\mathcal{D} \entails \mathcal{G}$, and
that we write $\dpth{\mathcal{D}}$ for the depth of the derivation
$\mathcal{D}$. The heavy lifting in the cut elimination proof is done
by the following lemma, which captures the intuition that cuts are
first shifted into the derivation of the left premiss of the cut until
the cut formula becomes principal there. Then they are shifted into
the derivation of the right premiss of the cut until they are
principal here as well and can be reduced to cuts on lower
complexity. The key idea is that because the calculus is end-active,
the cut formula essentially always occurs in the last component of one
of the premisses. As a technical subtlety, in order to shift up cuts
on the principal formula of the rule $\wbx_R^1$ or $\bbx_R^1$ we need
to remember that we can eliminate the occurrence of the cut formula in
the context. This is done by the additional conditions in the
statements $\SR_\wbx(n,m)$ and $\SR_\bbx(n,m)$ of the lemma, where we
use $\SL$ and $\SR$ as mnemonics for ``shift left'' and ``shift
right'', respectively, the latter with subscripts for the cut formula
being modal or propositional:

\begin{lemma}\label{lem:main-lem}
  The following statements hold for every $n,m$:
  \begin{description}
  \item[($\SR_\wbx(n,m)$)] Suppose that all of the following hold:
    \begin{itemize}
    \item $\mathcal{D}_1 \entails \mathcal{G} \fwd \Gamma \sa \Delta, \wbx
      A$ with $\wbx A$ principal in the last rule in $\mathcal{D}_1$
    \item $\mathcal{D}_2 \entails \mathcal{H} \fwd \wbx A, \Sigma \sa \Pi
      \upd \mathcal{I}$
    \item $\dpth{\mathcal{D}_1}+\dpth{\mathcal{D}_2} \leq m$
    \item there is a derivation of $\mathcal{G} \mrg \mathcal{H} \fwd
      \Gamma, \Sigma \sa \Delta, \Pi \fwd \epsilon\sa A$
    \item $\cmp{\wbx A}\leq n$.
    \end{itemize}
    Then there is a
    derivation of $\mathcal{G} \mrg \mathcal{H} \fwd \Gamma, \Sigma \sa
    \Delta, \Pi \upd \mathcal{I}$.
  \item[($\SR_\bbx(n,m)$)] Suppose that all of the following hold:
    \begin{itemize}
    \item $\mathcal{D}_1 \entails \mathcal{G} \bwd \Gamma \sa \Delta, \bbx
      A$ with $\bbx A$ principal in the last rule in $\mathcal{D}_1$
    \item $\mathcal{D}_2 \entails \mathcal{H} \bwd \bbx A, \Sigma \sa \Pi
      \upd \mathcal{I}$
    \item $\dpth{\mathcal{D}_1}+\dpth{\mathcal{D}_2} \leq m$
    \item there is a derivation of $ \mathcal{G} \mrg \mathcal{H} \bwd
      \Gamma, \Sigma \sa \Delta, \Pi \bwd \epsilon\sa A$
    \item $\cmp{\bbx A}\leq n$.
    \end{itemize}
    Then there is a
    derivation of $\mathcal{G} \mrg \mathcal{H} \fwd \Gamma, \Sigma \sa
    \Delta, \Pi \upd \mathcal{I}$.
  \item[($\SR_p(n,m)$)] Suppose that all of the following hold
    where $\mathcal{D}_1\upd \Gamma \sa \Delta, A$ and
      $\mathcal{H} \upd A, \Sigma \sa \Pi$ are structurally equivalent:
    \begin{itemize}
    \item $\mathcal{D}_1 \entails \mathcal{G} \upd \Gamma
      \sa \Delta, A$ with $A$ principal in the last applied rule in
      $\mathcal{D}_1$
    \item $\mathcal{D}_2 \entails \mathcal{H} \upd
      A, \Sigma \sa \Pi \upd \mathcal{I}$
    \item $\dpth{\mathcal{D}_1}+\dpth{\mathcal{D}_2} \leq m$
    \item $\cmp{A} \leq n$
    \item $A$ not of
    the form $\wbx B$ or $\bbx B$.
    \end{itemize}
    Then there is a derivation of $\mathcal{G} \mrg
    \mathcal{H} \upd \Gamma, \Sigma \sa \Delta, \Pi \upd \mathcal{I}$.
  \item[($\SL(n,m)$)] If $\mathcal{D}_1 \entails \mathcal{G} \upd \Gamma \sa
    \Delta, A \upd \mathcal{I}$ and $\mathcal{D}_2 \entails
    \mathcal{H} \upd A, \Sigma \sa \Pi$ with $\cmp{A} \leq n$ and
    $\dpth{\mathcal{D}_1}+ \dpth{\mathcal{D}_2} \leq m$, and
      $\mathcal{G} \upd \Gamma \sa \Delta$ and $\mathcal{H} \upd A,
      \Sigma \sa \Pi$ are structurally equivalent, then there is a derivation
    of $\mathcal{G} \mrg \mathcal{H} \upd \Gamma, \Sigma \sa \Delta,
    \Pi \upd \mathcal{I}$.
  \end{description}
\end{lemma}

The full proof is in the appendix.  As an immediate corollary,
using the statement $\SL(n,m)$ from Lem.~\ref{lem:main-lem} for
suitable $n,m$ we obtain:

\begin{theorem}[Cut elimination]
  Whenever $\entails\mathcal{G} \upd \Gamma \sa \Delta, A$ and
  $\entails \mathcal{H} \upd A, \Sigma \sa \Pi$, then also
  $\entails \mathcal{G} \mrg \mathcal{H} \upd \Gamma, \Sigma \sa
  \Delta, \Pi$.\qed
\end{theorem}

As usual, we will use cut elimination to show completeness.  However,
we also need to show admissibility of the \emph{necessitation rules}
$A / \wbx A$ and $A / \bbx A$. While this is straightforward in
standard calculi for modal logics, due to the additional premiss in
the rules $\wbx_R^1$ and $\bbx_R^1$, here we need to do some work:

\begin{theorem}[Admissibility of necessitation]\label{lem:nec}
  If $\epsilon\sa A$ is derivable in $\LNS_\Kt$, then so are
  $\epsilon\sa \wbx A$ and $\epsilon\sa \bbx A$.
\end{theorem}

\begin{proof}
  We consider the proof for $\epsilon\stt \wbx A$, the other case is
  analogous. To refer to problematic applications of the $\bbx_R^2$
  rule, we introduce some terminology.

  \begin{definition}
    Let $\mathcal{D}$ be the derivation of the sequent $\stt A$. An
    application $r$ of the rule $\bbx_R^2$ is \emph{critical} in
    $\mathcal{D}$ if its conclusion has exactly one component. The
    \emph{depth} of a critical application $r$ of $\bbx_R^2$ is the
    depth of the sub-derivation of $\mathcal{D}$ ending with this rule
    application, written $\dpth{r}$.
  \end{definition}

  Let $\mathcal{D}$ be a derivation of $\stt A$, and let
  $\crit{\mathcal{D}}$ be the set of critical applications of
  $\bbx_R^2$ in $\mathcal{D}$. For every possible depth $d$ of a
  critical application in $\crit{\mathcal{D}}$, fix an enumeration of
  all critical applications in $\crit{\mathcal{D}}$ with this
  depth. We then convert the derivation $\mathcal{D}$ bottom-up into a
  \emph{derivation from assumptions} of $\stt \wbx A$, i.e., a
  derivation of $\stt\wbx A$ where the leaves might be labelled with
  arbitrary linear nested sequents called the \emph{assumptions}. Each
  of these comes from one of the critical applications of
  $\mathcal{D}$, i.e., we have an injection $\iota$ from
  $\crit{\mathcal{D}}$ to the set of assumptions of the so far
  constructed derivation with assumptions. To each assumption $A$ we
  associate an \emph{index}, i.e., a triple $(d,i,c)$ of natural
  numbers, where $d$ is the depth of the critical application
  $\iota^{-1}(A)$, the number $i$ is the index of $\iota^{-1}(A)$ in
  the enumeration of critical applications of depth $d$, and
  $c \leq \dpth{\mathcal{D}}$ is a number corresponding to the depth
  of the current position in the original derivation $\mathcal{D}$. To
  ensure termination of the procedure, we consider the lexicographic
  ordering $<_{lex}$ on the indices $(d,i,c)$, and the \emph{multiset
    ordering} $\prec$ induced by $<_{lex}$ on the set of multisets of
  indices~\cite{Dershowitz:1979Multisets}. In particular for two such
  multisets $\mathcal{A}, \mathcal{B}$ we have that
  $\mathcal{A} \prec \mathcal{B}$ iff $\mathcal{B}$ can be obtained
  from $\mathcal{A}$ by replacing one or more indices $(d,i,c)$ by a
  finite number of indices $(d',i',c')$ with
  $(d',i',c') <_{lex} (d,i,c)$. It is shown in \emph{op.cit.}\ that
  $\prec$ is well-founded.

  The first ingredient in the construction of the derivation of
  $\stt \wbx A$ is given by essentially prefixing
  $\epsilon\stt \wbx A$ to every linear nested sequent in
  $\mathcal{D}$:

  \begin{definition}
    Let $\mathcal{E}$ be a sub-derivation of $\mathcal{D}$ and
    $\Gamma \stt \Delta$ a sequent. For any natural number $n$ the
    derivation $(\Gamma \stt \Delta) \fwd \mathcal{E}(n)$ is obtained
    by prefixing $\Gamma \stt \Delta \fwd$ to every linear nested
    sequent in $\mathcal{E}$, and replacing critical applications of
    $\bbx_R^2$ with applications of $\bbx_R^1$ and an assumption as
    follows:
    \[
      \begin{array}{l}
        \infer[\bbx_R^2]{\Sigma \stt \Pi, \bbx B
        }
        {\Sigma \stt \Pi, \bbx B \bwd \epsilon\stt B
        }\smallskip\\
        \leadsto\quad
        \vcenter{
        \infer[\bbx_R^1]{\Gamma\stt\Delta \fwd \Sigma \stt \Pi, \bbx B
        }
        {\infer[\EW]{\Gamma\stt\Delta, B \fwd \Sigma \stt \Pi, \bbx B
        }
        {\Gamma\stt \Delta, B
        }
        &
          \Gamma\stt \Delta \fwd \Sigma \stt \Pi, \bbx B \bwd \epsilon
          \stt B
          }
          }
      \end{array}
    \]
    The index $(d,i,n)$ of the assumption $\Gamma \stt \Delta,B$ is
    given by the depth $d$ of the original critical application of
    $\bbx_R^2$, its index $i$, and the number $n$.
  \end{definition}

  In the first step we obtain from $\mathcal{D}$ the derivation with
  assumptions
  $(\epsilon\stt \wbx A) \fwd \mathcal{D}(\dpth{\mathcal{D}})$. The
  conclusion of this derivation is
  $\epsilon\stt \wbx A \fwd \epsilon\stt A$, hence applying $\wbx_R^2$
  we will ultimately obtain a derivation with assumptions of
  $\epsilon\stt \wbx A$.
  
  The next step is to construct a derivation for each assumption,
  starting with one of maximal index. The general idea is to copy the
  derivation of the premiss of the corresponding critical application
  of $\bbx_R^2$, but essentially ``folding back'' the second component
  of the original derivation into the first one of the new derivation
  until the linear nested sequents in the original derivation are
  reduced to one component again. This means that the first component
  of the new derivation will collect a number of second components
  occurring in the original derivation. To make this precise, for a
  sequent $\Omega \stt \Theta$, a derivation $\mathcal{E}$ with
  assumptions, a critical rule application $r$ and a natural number
  $n$, we write
  $(\Omega \stt \Theta) \mrg \mathcal{E}(r \leftarrow n)$ for the
  derivation with assumptions obtained from $\mathcal{E}$ by merging
  the first component of each linear nested sequent in $\mathcal{E}$
  with the sequent $\Omega \stt \Theta$, and changing the indices
  $(d,i,c)$ of all those assumption in $\mathcal{E}$ corresponding to
  $r$ to $(d,i,n)$.

  Take an assumption $\Gamma \stt \Delta, B$ with index $(d,i,c)$
  which is maximal w.r.t.\ $<_{lex}$, and suppose that the
  corresponding critical rule application $r$ is given by:
  \[
    \infer[\bbx_R^2]{\Sigma \stt \Pi, \bbx B
    }
    {\infer*[\mathcal{E}]{\Sigma \stt \Pi, \bbx B \bwd \epsilon\stt B}{}
    }
  \]
  Suppose that the assumption occurs in the context
  \[
    \infer*[\mathcal{F}]{
    }
    {\infer[\bbx_R^1]{\Gamma \stt \Delta \fwd \Sigma \stt \Pi, \bbx B
      }
      {\infer[\EW]{\Gamma \stt \Delta, B \fwd \Sigma \stt \Pi, \bbx B
        }
        {\Gamma \stt \Delta, B
        }
        &
        \infer*[]{\Gamma \stt \Delta \fwd \Sigma \stt \Pi, \bbx B \fwd
          \epsilon\stt B
        }
        {
        }
      }
    }
  \]
  where $\mathcal{F}$ is the derivation with assumptions below the
  conclusion of the application of $\bbx_R^1$.  Note that all
  assumptions in $\mathcal{F}$ have index smaller than $(d,i,c)$.  We
  extend this derivation upwards by applying the same rules as in the
  original derivation $\mathcal{E}$, until in $\mathcal{E}$ we
  encounter a rule $\wbx_L^2$ or $\EW$ which shortens the sequent to
  only the first component again. This is straightforward unless in
  the original derivation we have an application of a rule in which
  the first component is active, i.e., an application of the rules
  $\bbx_L^1$ or $\wbx_R^1$ with active first component.

  The case of $\wbx_R^1$ is unproblematic, replacing $\wbx_R^1$ with
  $\wbx_R^2$ and continuing upwards as in the derivation of the right
  premiss.  Note that the first component in the original derivation
  stays the same.
  
  In the case of an application of $\bbx_L^1$ we recreate the original
  first component $\Sigma \stt \Pi, \bbx B$ using $\mathcal{F}$. In
  general, this creates new copies of the assumptions in
  $\mathcal{F}$, in particular of other assumptions corresponding to
  $r$.  To ensure termination we decrease the index of every
  assumption corresponding to $r$ to the depth of the current position
  in the original derivation. Hence the multiset of assumptions of the
  new derivation is smaller than that of the old one w.r.t.\ $\prec$.
  Suppose that we encounter an application of the rule $\bbx_L^1$ in
  the form
  \[
    \infer[\bbx_L^1]{\Sigma', \bbx C \stt \Pi', \bbx B \bwd \Xi
      \stt \Upsilon
    }
    {\infer*[\mathcal{G}]{\Sigma', \bbx C \stt \Pi', \bbx B,
         \bwd \Xi, C
      \stt \Upsilon
      }
      {
      }
    }
  \]
  Since all linear nested sequents between the conclusion of this rule
  application and the critical rule application $r$ contain at least
  two components, and since when simulating applications of $\wbx_R^1$
  as above we never changed the first component, the first component
  $\Sigma', \bbx C \stt \Pi', \bbx B$ stays the same as the original
  first component $\Sigma \stt \Pi, \bbx B$. Hence we can recreate
  this component and continue as:
  \[
    \infer*[(\Xi \stt \Upsilon)\oplus \mathcal{F} (r \leftarrow
    \dpth{\mathcal{G}})]{\Gamma,\Xi \stt \Delta,B,\Upsilon
    }
    {\infer[\bbx_L^2]{\Gamma,\Xi \stt \Delta,B,\Upsilon \fwd \Sigma',
        \bbx C \stt \Pi',
        \bbx B
      }
      {\Gamma, \Xi, C \stt \Delta, B, \Upsilon
      }
    }
  \]
  Continuing upwards like this, in the original derivation we
  eventually reach initial sequents, or applications of $\wbx_L^2$ or
  $\EW$ which reduce the number of components to one. In the latter
  case, we again recreate the original first component. E.g., suppose
  that in the original derivation we have an application of $\wbx_L^2$
  in the form
  \[
    \infer[\wbx_L^2]{\Sigma' \stt \Pi' \bwd \Xi, \wbx C \stt \Upsilon
    }
    {\infer*[\mathcal{G}]{\Sigma', C \stt \Pi'
      }
      {
      }
    }
  \]
  Then again we have that $\Sigma' \stt \Pi'$ is the same as the first
  component $\Sigma \stt \Pi, \bbx B$ of the critical rule application
  $r$, and hence we can recreate it and continue using
  \[
    \infer*[(\Xi \stt \Upsilon) \oplus \mathcal{F}(r \leftarrow
    \dpth{\mathcal{G}})]{\Gamma, \Xi, \wbx C \stt \Delta, \Upsilon
    }
    {\infer[\wbx_L^1]{\Gamma, \Xi, \wbx C \stt \Delta, \Upsilon \fwd
        \Sigma' \stt \Pi' 
      }
      {\infer*[(\Gamma,\Xi,\wbx C \stt \Delta, \Upsilon) \fwd
        \mathcal{G}(\dpth{\mathcal{G}})]{\Gamma, \Xi, \wbx C \stt
          \Delta, \Upsilon \fwd
        \Sigma', C \stt \Pi' 
        }
        {
        }
      }
    }
  \]
  Note that again the multiset of indices of assumptions is decreased
  wrt.~$\prec$. In particular, the depth of every critical rule
  application in $\mathcal{G}$ is smaller than the depth of the
  critical rule application $r$. The case for the rule $\EW$ is
  analogous.

  Continuing in this way we replace every assumption by a finite
  multiset of smaller ones. Hence the sequence of multisets of
  assumptions is strictly decreasing wrt.\ the well-ordering $\prec$,
  and the procedure must terminate. When it does we obtain a
  derivation without assumptions, giving a derivation of
  $\epsilon\stt \wbx A$. \qed
\end{proof}

\begin{theorem}[Completeness]
  The system $\LNS_\Kt$ is cut-free complete for $\Kt$.
\end{theorem}

\begin{proof}
  It is straightforward to derive the axioms.  Modus ponens is
  simulated as usual using cuts. The necessitation rules are simulated
  using Lem.~\ref{lem:nec}.\qed
\end{proof}

\section{Application: Linear nested sequents for modal logic $\KB$}

It is rather straightforward to adapt our system to capture modal
logic $\KB$. Semantically, $\KB$ is given as the mono-modal logic of
\emph{symmetric Kripke frames}, i.e., frames with symmetric
accessibility relation.  Syntactically, $\KB$ is obtained from $\Kt$
by collapsing the forwards and backwards modalities, e.g., via adding
the axiom $\wbx A \leftrightarrow \bbx A$. Correspondingly, we also
collapse the structural connectives $\fwd$ and $\bwd$ to obtain the
simpler definition of linear nested sequents for $\KB$ via the grammar
$S := \Gamma \stt \Delta \mid \Gamma \stt \Delta \fwd S$. The simplest
version of the linear nested sequent calculus $\LNS_\KB$ for modal
logic $\KB$ then contains the propositional rules and rule $\EW$ of
Fig.~\ref{fig:rules-kb} together with the two standard rules
\[
  \infer[\wbx_R]{\mathcal{G} \fwd \Gamma \stt \Delta, \wbx A
  }
  {\mathcal{G} \fwd \Gamma \stt \Delta, \wbx A \fwd \epsilon \stt A
  }
  \qquad
  \infer[\wbx_L^1]{\mathcal{G} \fwd \Gamma, \wbx A \stt \Delta \fwd
    \Sigma \stt \Pi
  }
  {\mathcal{G} \fwd \Gamma, \wbx A \stt \Delta \fwd
    \Sigma, A \stt \Pi
  }
\]
found in (linear) nested sequent calculi for modal logic $\K$ and the
single new rule
\[
  \infer[\wbx_L^2]{\mathcal{G} \fwd \Gamma \stt \Delta \fwd \Sigma, \wbx A
    \stt \Pi
  }
  {\mathcal{G} \fwd \Gamma,A \stt \Delta
  }
\]
Soundness is seen analogously to Thm.~\ref{thm:soundness}, and
completeness follows by repeating the proofs for $\Kt$, at each step
collapsing the forwards and backwards modalities:

\begin{theorem}
  The calculus $\LNS_\KB$ is sound and complete for modal logic
  $\KB$.\qed
\end{theorem}

In comparison with the linear nested sequent calculus for modal logic
$\KB$ introduced by Parisi~\cite{Parisi:2017PhD}, we do not need to
change the direction of the linear nested sequents, and (a variant of)
our system has syntactic cut elimination. Note also that the system
$\LNS_\KB$ is essentially the end-active and linear version of the
nested sequent calculus for $\KB$ of Br\"unnler and
Poggiolesi~\cite{Brunnler:2009kx,Poggiolesi:2010} with the crucial
difference that the last component is deleted in the premiss of the
symmetry rule $\wbx_L^2$. Since derivations of $\LNS_\KB$ can be
transformed straightforwardly bottom-up into derivations in the full
nested sequent system considered in \emph{op.~cit.}, our completeness
result implies the completeness results there.

\section{Conclusion}

We have seen that linear nested sequents are so far the minimal
extension of traditional sequents needed to handle tense logics and
modal logic $\KB$. Intuitively, they provide the semantic expressive
power to look both ways along the underlying Kripke reachability
relation while also providing a rigorous and modular proof-theoretic
framework. The main novelty to mimic traditional tableau calculi for
tense logics is the addition of restart rules to maintain the
compatibility between parent nodes and their children.

In future work we would like to explore the possibility of extending
our calculus to capture further properties of the accessibility
relation such as reflexivity, forwards or backwards directedness, or
transitivity. We conjecture that suitable modifications of the rules
$\wbx_R^1$ and $\bbx_R^1$ in the spirit of the ones presented here
should suffice for a cut elimination proof. It is perhaps less obvious
that the proof of admissibility of necessitation goes through in these
cases as well. Finally, we would like to investigate whether it is
possible to use our calculi in complexity-optimal decision procedures.

\newpage

\bibliographystyle{splncs04}
\bibliography{linearnestedtenselogics}

\newpage
\appendix

\section{Additional Proofs}
\label{sec:appendix}

\setcounter{theorem}{15}
\begin{lemma}
  The following statements hold for every $n,m$:
  \begin{enumerate}
  \item[$(\SR_\wbx(n,m))$] Suppose that all of the following hold:
    \begin{itemize}
    \item $\mathcal{D}_1 \entails \mathcal{G} \fwd \Gamma \sa \Delta, \wbx
      A$ with $\wbx A$ principal in the last rule in $\mathcal{D}_1$
    \item $\mathcal{D}_2 \entails \mathcal{H} \fwd \wbx A, \Sigma \sa \Pi
      \upd \mathcal{I}$
    \item $\dpth{\mathcal{D}_1}+\dpth{\mathcal{D}_2} \leq m$
    \item there is a derivation of $\mathcal{G} \mrg \mathcal{H} \fwd
      \Gamma, \Sigma \sa \Delta, \Pi \fwd \epsilon\sa A$
    \item $\cmp{\wbx A}\leq n$.
    \end{itemize}
    Then there is a
    derivation of $\mathcal{G} \mrg \mathcal{H} \fwd \Gamma, \Sigma \sa
    \Delta, \Pi \upd \mathcal{I}$.
  \item[$(\SR_\bbx(n,m))$] Suppose that all of the following hold:
    \begin{itemize}
    \item $\mathcal{D}_1 \entails \mathcal{G} \bwd \Gamma \sa \Delta, \bbx
      A$ with $\bbx A$ principal in the last rule in $\mathcal{D}_1$
    \item $\mathcal{D}_2 \entails \mathcal{H} \bwd \bbx A, \Sigma \sa \Pi
      \upd \mathcal{I}$
    \item $\dpth{\mathcal{D}_1}+\dpth{\mathcal{D}_2} \leq m$
    \item there is a derivation of $ \mathcal{G} \mrg \mathcal{H} \bwd
      \Gamma, \Sigma \sa \Delta, \Pi \bwd \epsilon\sa A$
    \item $\cmp{\bbx A}\leq n$.
    \end{itemize}
    Then there is a
    derivation of $\mathcal{G} \mrg \mathcal{H} \fwd \Gamma, \Sigma \sa
    \Delta, \Pi \upd \mathcal{I}$.
  \item[$(\SR_p(n,m))$] Suppose that all of the following hold:
    \begin{itemize}
    \item $\mathcal{D}_1 \entails \mathcal{G} \upd \Gamma
      \sa \Delta, A$ with $A$ principal in the last applied rule in
      $\mathcal{D}_1$
    \item $\mathcal{D}_2 \entails \mathcal{H} \upd
      A, \Sigma \sa \Pi \upd \mathcal{I}$
    \item $\dpth{\mathcal{D}_1}+\dpth{\mathcal{D}_2} \leq m$
    \item $\cmp{A} \leq n$
    \item $A$ not of
    the form $\wbx B$ or $\bbx B$.
    \end{itemize}
    Then there is a derivation of $\mathcal{G} \mrg
    \mathcal{H} \upd \Gamma, \Sigma \sa \Delta, \Pi \upd \mathcal{I}$.
  \item[$(\SL(n,m))$] If
    $\mathcal{D}_1 \entails \mathcal{G} \upd \Gamma \sa \Delta, A \upd
    \mathcal{I}$ and
    $\mathcal{D}_2 \entails \mathcal{H} \upd A, \Sigma \sa \Pi$ with
    $\cmp{A} \leq n$ and
    $\dpth{\mathcal{D}_1}+ \dpth{\mathcal{D}_2} \leq m$, then there is
    a derivation of
    $\mathcal{G} \mrg \mathcal{H} \upd \Gamma, \Sigma \sa \Delta, \Pi
    \upd \mathcal{I}$.
  \end{enumerate}
\end{lemma}

\begin{proof}
  We prove all four statements simultaneously by induction on the
  tuples $(n,m)$ in the lexicographic ordering. The step case for
  $\SR_\wbx(n,m)$ makes use of $\SR_\wbx(n,m-1)$, $\SL(n,m-1)$ and
  $\SL(n-1,m)$. Analogously for the case for $\SR_\bbx(n,m)$. For
  $\SR_p(n,m)$ we use $\SL(n,m-1)$, $\SL(n-1,k)$ and
  $\SR_p(n,m-1)$. The case for $\SL(n,m)$ uses the statements
  $\SL(n,m-1)$, $\SR_\wbx(n,m)$, $\SR_\bbx(n,m)$ and $\SR_p(n,m)$.

\subsubsection{Cases for $\SR_\wbx(n,m)$}
\label{sec:cases-srl_1}

\paragraph{Case: principal $\wbx_R^1$ vs principal $\wbx_L^1$.}
In this case the derivations end in:
\[
  \infer[\wbx_R^1]{\mathcal{G} \upd \Gamma \sa \Delta \bwd \Sigma
    \sa \Pi, \wbx A
  }
  {\infer*{\mathcal{G} \upd \Gamma \sa \Delta,A \bwd \Sigma
    \sa \Pi, \wbx A}{\mathcal{D}_3}
    &
    \infer*{\mathcal{G} \upd \Gamma \sa \Delta \bwd \Sigma
      \sa \Pi, \wbx A \fwd \epsilon\sa A
    }
    {\mathcal{D}_4
    }
  }
\]
and
\[
  \infer[\wbx_L^1]{\mathcal{H}\upd \Gamma' \sa \Delta' \bwd \Sigma',
    \wbx A \sa \Pi' \fwd \Omega \sa \Theta
  }
  {\infer*{\mathcal{H}\upd \Gamma' \sa \Delta' \bwd  \Sigma',
    \wbx A \sa \Pi' \fwd \Omega, A \sa \Theta}{\mathcal{D}_5}
  }
\]
By $\SR_\wbx(n,m-1)$ on the conclusion of $\wbx_R^1$ and the premiss of
$\wbx_L^1$ we obtain a derivation $\mathcal{D}_6$ of
\[
  \mathcal{G} \mrg \mathcal{H} \upd \Gamma, \Gamma' \sa \Delta,
  \Delta' \bwd \Sigma, \Sigma' \sa \Pi, \Pi' \fwd \Omega, A \sa \Theta\;.
\]
Note that we can apply $\SR_\wbx(n,m-1)$, because by assumption we
know that there is a derivation $\mathcal{D}_7$ of
\[
  \mathcal{G} \mrg \mathcal{H} \upd \Gamma, \Gamma' \sa
  \Delta,\Delta' \bwd \Sigma,\Sigma' \sa \Pi,\Pi' \fwd \epsilon\sa A\;.
\]
Further, applying $\SL(n-1,\dpth{\mathcal{D}_6} + \dpth{\mathcal{D}_7})$ to
these two linear nested sequents yields a derivation of
\[
  \mathcal{G} \mrg \mathcal{H}\mrg \mathcal{G} \mrg \mathcal{H} \fwd
  \Gamma, \Gamma',\Gamma, \Gamma' \sa
  \Delta,\Delta',\Delta,\Delta' \fwd \Sigma,\Sigma',\Sigma,\Sigma'
  \sa \Pi,\Pi',\Pi,\Pi' \fwd \Omega \sa \Theta
\]
Now admissibility of contraction yields the desired
\[
  \mathcal{G} \mrg \mathcal{H} \fwd \Gamma,\Gamma' \sa \Delta,
  \Delta' \fwd \Sigma,\Sigma' \sa \Pi,\Pi' \fwd \Omega \sa \Theta\;.
\]

\paragraph{Case: principal $\wbx_R^1$ vs principal $\wbx_L^2$.}
In this case the derivations end in:
\[
  \infer[\wbx_R^1]{\mathcal{G} \upd \Gamma \sa \Delta \bwd \Sigma
    \sa \Pi, \wbx A
  }
  {\infer*{\mathcal{G} \upd \Gamma \sa \Delta,A \bwd \Sigma
    \sa \Pi, \wbx A}{\mathcal{D}_3}
    &
    \infer*{\mathcal{G} \upd \Gamma \sa \Delta \bwd \Sigma
      \sa \Pi, \wbx A \fwd \epsilon\sa A
    }
    {\mathcal{D}_4
    }
  }
\]
and
\[
  \infer[\wbx_L^2]{\mathcal{H}\upd \Gamma' \sa \Delta' \bwd \Sigma',
    \wbx A \sa \Pi' 
  }
  {\infer*{\mathcal{H}\upd \Gamma',A \sa \Delta' }{\mathcal{D}_5}
  }
\]
An application of $\SL(n,m-1)$ to the left premiss of $\wbx_R^1$ and
the conclusion of $\wbx_L^2$ yields a derivation $\mathcal{D}_6$ of 
\[
  \mathcal{G} \mrg \mathcal{H} \upd \Gamma, \Gamma' \sa \Delta, A,
  \Delta' \bwd \Sigma,\Sigma' \sa \Pi,\Pi'
\]
Now an application of
$\SL(n-1,\dpth{\mathcal{D}_6}+\dpth{\mathcal{D}_5})$ gives
\[
  \mathcal{G} \mrg \mathcal{H} \mrg \mathcal{H} \upd \Gamma,\Gamma',
  \Gamma' \sa \Delta,\Delta',\Delta' \bwd \Sigma,\Sigma' \sa \Pi,\Pi'
\]
and contraction yields the desired result.

\paragraph{Case: principal $\wbx_R^1$ vs contextual $\wbx_L^2$.}
In case the premiss of $\wbx_L^2$ is not shorter than the conclusion
of $\wbx_R^1$, we simply apply $\SR_\wbx(n,m-1)$ to the conclusion of
$\wbx_R^1$ and the premiss of $\wbx_L^2$, followed by $\wbx_L^2$. The
additional assumption in $\SR_\wbx(n,m-1)$ of existence of a
derivation is trivially satisfied, because we can use the same
derivation we have by assumption. If the premiss of $\wbx_L^2$ is
shorter than the conclusion of $\wbx_R^1$ the derivations end in:
\[
  \infer[\wbx_R^1]{\mathcal{G} \upd \Gamma \sa \Delta \bwd \Sigma
    \sa \Pi, \wbx A
  }
  {\infer*{\mathcal{G} \upd \Gamma \sa \Delta,A \bwd \Sigma
    \sa \Pi, \wbx A}{\mathcal{D}_3}
    &
    \infer*{\mathcal{G} \upd \Gamma \sa \Delta \bwd \Sigma
      \sa \Pi, \wbx A \fwd \epsilon\sa A
    }
    {\mathcal{D}_4
    }
  }
\]
and
\[
  \infer[\wbx_L^2]{\mathcal{H}\upd \Gamma' \sa \Delta' \bwd \Sigma',
    \wbx A, \wbx B \sa \Pi' 
  }
  {\infer*{\mathcal{H}\upd \Gamma',B \sa \Delta' }{\mathcal{D}_5}
  }
\]
We simply replace the application of $\wbx_L^2$ with
\[
  \infer[\wbx_L^2]{\mathcal{H}\upd \Gamma' \sa \Delta' \bwd \Sigma',
    \wbx B \sa \Pi' 
  }
  {\infer*{\mathcal{H}\upd \Gamma',B \sa \Delta' }{\mathcal{D}_5}
  }
\]
Then by admissibility of internal weakening we obtain $\mathcal{G}
\mrg \mathcal{H} \upd \Gamma, \Gamma' \sa \Delta,\Delta' \bwd
\Sigma,\Sigma', \wbx B \sa \Pi,\Pi'$ as desired.

\paragraph{Case: principal $\wbx_R^1$ vs $\EW$.}
We have two subcases, depending on whether the premiss of the
application of $\EW$ is shorter than the conclusion of $\mathcal{D}_1$
or not. If it is not shorter, we apply $\SR_\wbx(n,m-1)$ to the
premiss of $\EW$, followed by the same rule. If it is shorter, the two
derivations $\mathcal{D}_1$ and $\mathcal{D}_2$ end in
\[
  \mathcal{G} \upd \Gamma \sa \Delta \bwd \Sigma \sa \Pi, \wbx A
\]
and
\[
  \infer[\EW]{\mathcal{H} \upd \Gamma' \sa \Delta' \bwd \wbx A,\Sigma' \sa \Pi'
  }
  {\infer*{\mathcal{H} \upd \Gamma' \sa \Delta'}{\mathcal{D}_3}
  }
\]
respectively. We replace this last application of $\EW$ with
\[
  \infer[\EW]{\mathcal{H} \upd \Gamma,\Sigma \sa \Delta,\Pi
  }
  {\infer*{\mathcal{H}}{\mathcal{D}_3}
  }
\]
and use admissibility of weakening to obtain the desired $\mathcal{G} \mrg
\mathcal{H} \upd \Gamma,\Sigma \sa \Delta,\Pi$.

\paragraph{Case: principal $\wbx_R^1$ vs context in other rules.}
The rules $\wbx_L^2$, $\bbx_L^2$ and $\EW$ are the only rules in which
the premiss is shorter than the conclusion, and the cases of
$\wbx_L^2$ and $\EW$ were covered above. For the case of $\bbx_L^2$,
since the linear nested sequents need to be structurally equivalent up
to the component containing the cut formula, the cut formula cannot be
in the last component of the conclusion of $\bbx_L^2$. Because of
this, and since no rule removes any formulae when moving from
conclusion to premisses, in this case the derivation $\mathcal{D}_2$
must end in
\[
  \infer[R]{\mathcal{H} \upd \Gamma' \sa \Delta' \bwd \wbx A, \Sigma'
    \sa \Pi' \upd \mathcal{I}
  }
  {\infer*{\mathcal{H} \upd \Gamma' \sa \Delta' \bwd \wbx A, \Sigma''
      \sa \Pi'' \upd \mathcal{I}'}{\mathcal{D}_3'}
  }
\]
for some rule $R$ with $\Sigma'' \supset \Sigma'$ and
$\Pi''\supset \Pi'$. By assumption we know that there is a derivation
of
$\mathcal{G} \mrg \mathcal{H} \upd \Gamma, \Gamma' \sa \Delta, \Delta'
\bwd \Sigma,\Sigma' \sa \Pi,\Pi' \fwd \epsilon\sa A$, hence by
admissibility of weakening we have a derivation of
$\mathcal{G} \mrg \mathcal{H} \upd \Gamma, \Gamma' \sa \Delta, \Delta'
\bwd \Sigma,\Sigma'' \sa \Pi,\Pi'' \fwd \epsilon\sa A$. Thus applying
$\SR_\wbx(n,m-1)$ to the derivation $\mathcal{D}_1$ and the premiss of
the rule $R$ yields
\[
  \mathcal{G} \mrg \mathcal{H} \fwd \Gamma, \Gamma' \sa \Delta,\Delta'
  \bwd \Sigma,\Sigma'' \sa \Pi, \Pi''
  \upd \mathcal{I}'
\]
and an application of $R$ gives
$\mathcal{G} \mrg \mathcal{H} \upd \Gamma, \Gamma' \sa \Delta,\Delta'
\bwd \Sigma,\Sigma' \sa \Pi, \Pi' \upd \mathcal{I}$.

\paragraph{Case: principal $\wbx_R^2$ vs principal $\wbx_L^1$.}
As for the case principal $\wbx_R^1$ vs principal $\wbx_L^1$.

\paragraph{Case: principal $\wbx_R^2$ vs principal $\wbx_L^2$.}
This case cannot occur, since the conclusions would not be
structurally equivalent.

\paragraph{Case: principal vs $\idrl$ or $\bot_L$.}
Since no logical rule has a propositional variable or $\bot$ as
principal formula, the formula $A$ must be part of the context in
$\idrl$ resp. $\bot_L$. Hence the desired linear nested sequent also
is the conclusion of an application of $\idrl$ resp. $\bot_L$.

\subsubsection{Cases for $\SR_\bbx$}
\label{sec:cases-sr_bbx}

Analogous to the cases for $\SR_\wbx$, with $\wbx$ and $\bbx$
inverted, as well as $\fwd$ and $\bwd$.

\subsubsection{Cases for $\SR_p$}
\label{sec:cases-sr_2}

\paragraph{Case: principal $\to_R$ vs principal $\to_L$.}
As usual: apply \emph{cross cuts}, i.e., applications of $\SL(n,m-1)$
to the conclusion $\to_R$ and the premisses of $\to_L$ and vice versa
to eliminate the occurrences of the principal formula from the
premisses. Then apply $\SL(k,\ell)$ with $k<n$ on the resulting
derivations to eliminate the auxiliary formulae, followed by
admissibility of contraction.

\paragraph{Case: principal vs contextual in $\wbx_L^2$ or $\bbx_L^2$.}
Same as the corresponding cases in the proofs of $\SR_\wbx(n,m)$ and
$\SR_\bbx(n,m)$ respectively.

\paragraph{Case: principal vs $\EW$.}
Again, we have two subcases, depending on whether the premiss of the
application of $\EW$ is shorter than the conclusion of $\mathcal{D}_1$
or not. If it is not shorter, we apply $\SR_p(n,m-1)$ to the premiss
of $\EW$, followed by the same rule. If it is shorter, the two
derivations $\mathcal{D}_1$ and $\mathcal{D}_2$ end in
\[
  \mathcal{G} \upd \Gamma \sa \Delta, A
\]
and
\[
  \infer[\EW]{\mathcal{H} \upd A,\Sigma \sa \Pi
  }
  {\infer*{\mathcal{H}}{\mathcal{D}_3}
  }
\]
respectively. We replace this last application of $\EW$ with
\[
  \infer[\EW]{\mathcal{H} \upd\Sigma \sa \Pi
  }
  {\infer*{\mathcal{H}}{\mathcal{D}_3}
  }
\]
and use admissibility of weakening to obtain the desired
$\mathcal{G} \mrg \mathcal{H} \fwd \Gamma,\Sigma \sa \Delta,\Pi$.

\paragraph{Case: principal vs context.}
Since the rules $\wbx_L^2, \bbx_L^2$ and $\EW$ are the only rules in
which the premiss is shorter than the conclusion, and since no rule
removes any formulae when moving from conclusion to premisses, in this
case the derivation $\mathcal{D}_2$ must end in
\[
  \infer[R]{\mathcal{H} \upd A, \Sigma \sa \Pi \upd \mathcal{I}
  }
  {\infer*{\mathcal{H} \upd A, \Sigma' \sa \Pi' \upd
      \mathcal{I}'}{\mathcal{D}_3'}
  }
\]
for some rule $R$. Now applying $\SR_p(n,m-1)$ to the derivation
$\mathcal{D}_1$ and the premiss of the rule $R$ yields
\[
  \mathcal{G} \mrg \mathcal{H} \upd \Gamma, \Sigma' \sa \Delta, \Pi'
  \upd \mathcal{I}'
\]
and an application of $R$ gives $\mathcal{G} \mrg \mathcal{H} \upd
\Gamma, \Sigma \sa \Delta, \Pi \upd \mathcal{I}$.

\paragraph{Case: principal vs $\idrl$ or $\bot_L$.}
Since no logical rule has a propositional variable or $\bot$ as
principal formula, the formula $A$ must be part of the context in
$\idrl$ resp. $\bot_L$. Hence the desired linear nested sequent also
is the conclusion of an application of $\idrl$ resp. $\bot_L$, followed
by $\EW$.

\subsubsection{Cases for $\SL$}
\label{sec:cases-sl}

\paragraph{Case: $A$ is principal in the last rule in
  $\mathcal{D}_1$.}
Since the principal formulae of all right rules are in the last
component, in this case the derivation $\mathcal{D}_1$ must end in
\[
  \infer*{\mathcal{G} \upd \Gamma \sa \Delta,A}{\mathcal{D}_1}
\]
We now distinguish cases according to the shape of $A$. If $A$ is not
of the shape $\wbx B$ or $\bbx B$, we apply $\SR_p(n,m)$ to obtain the
result. If $A$ is of the shape $\wbx B$, the last rule in
derivation~$\mathcal{D}_1$ is the rule $\wbx_R^1$ or $\wbx_R^2$. In
the first case, it ends in
\[
  \infer[\wbx_R^1]{\mathcal{G} \upd \Omega \sa \Theta \bwd \Gamma \sa \Delta, \wbx B
  }
  {\infer*{\mathcal{G} \upd \Omega \sa \Theta, B \bwd \Gamma \sa
      \Delta, \wbx B}{\mathcal{D}_3}
    &
    \infer*{\mathcal{G} \upd \Omega \sa \Theta \bwd \Gamma \sa \Delta, \wbx
    B \fwd \epsilon\sa B}{\mathcal{D}_4}
  }
\]
and $\mathcal{D}_2$ ends in
\[
  \infer*{\mathcal{H} \upd \Xi \sa \Upsilon \bwd \wbx B, \Sigma \sa \Pi}{\mathcal{D}_2}
\]
By $\SL(n,m-1)$ on $\mathcal{D}_4$ and $\mathcal{D}_2$ we know that
there is a derivation $\mathcal{D}_3$ of
$\mathcal{G}\mrg \mathcal{H} \upd \Omega,\Xi \sa \Theta,\Upsilon \bwd
\Gamma, \Sigma \sa \Delta, \Pi \fwd \epsilon\sa B$. Hence
$\SR_\wbx(n,m)$ is applicable and yields a derivation of
$\mathcal{G} \mrg \mathcal{H} \upd \Omega,\Xi \sa \Theta,\Upsilon \bwd
\Gamma, \Sigma \sa \Delta, \Pi$. Note that the depth of the derivation
$\mathcal{D}_3$ is irrelevant.

In case the last rule in $\mathcal{D}_1$ was $\wbx_R^2$, the reasoning
is the same.

If $A$ is of the shape $\bbx B$, the argument is analogous to the
above, using $\SR_\bbx(n,m)$ instead of $\SR_\wbx(n,m)$.

\paragraph{Case: Last rule in $\mathcal{D}_1$ is $\wbx_L^2$ or $\bbx_L^2$.}
We only consider the case of $\wbx_L^2$, the case of $\bbx_L^2$ is
analogous.  We distinguish cases according to whether the occurrence
of $A$ is in the last component or not. If it is not, we apply
$\SL(n,m-1)$ to the premiss of the application of $\EW$, followed by
the same rule. If the occurrence of $A$ is in the last component, the
derivation $\mathcal{D}_1$ ends in
\[
  \infer[\wbx_L^2]{\mathcal{G} \upd \Gamma \sa \Delta \bwd \Sigma,
    \wbx B \sa \Pi, A
  }
  {\infer*{\mathcal{G} \upd \Gamma, B \sa \Delta }{\mathcal{D}_3}
  }
\]
In this case we change the application of $\wbx_L^2$ to
\[
  \infer[\wbx_L^2]{\mathcal{G} \upd \Gamma \sa \Delta \bwd \Sigma,
    \wbx B \sa \Pi
  }
  {\infer*{\mathcal{G} \upd \Gamma, B \sa \Delta }{\mathcal{D}_3}
  }
\]
and are done using admissibility of internal weakening.

\paragraph{Case: last rule in $\mathcal{D}_1$ is $\EW$.}
We distinguish cases according to whether the occurrence of $A$ is in
the last component or not. If it is not, we apply $\SL(n,m-1)$ to the
premiss of the application of $\EW$, followed by the same rule. If the
occurrence of $A$ is in the last component, the derivation
$\mathcal{D}_1$ ends in
\[
  \infer[\EW]{\mathcal{G} \upd \Gamma \sa \Delta, A
  }
  {\infer*{\mathcal{G}}{\mathcal{D}_1'}
  }
\]
We replace this application of $\EW$ with
\[
  \infer[\EW]{\mathcal{G} \upd \Gamma,\Sigma \sa \Delta,\Pi
  }
  {\infer*{\mathcal{G}}{\mathcal{D}_1'}
  }
\]
and then admissibility of weakening yields the desired result.

\paragraph{Case: Last rule in $\mathcal{D}_1$ is $\idrl$ or
  $\bot_L$.}
In this case the desired linear nested sequent also is the conclusion
of $\idrl$ resp. $\bot_L$ followed by $\EW$.

\paragraph{Case: The formula $A$ is contextual in the last rule in
  $\mathcal{D}_1$ which is not $\wbx_L^2$ or $\bbx_L^2$.}
In this case $\mathcal{D}_1$ ends in
\[
  \infer[R]{\mathcal{G} \upd \Gamma \sa \Delta, A \upd \mathcal{I}
  }
  {\infer*{\mathcal{G} \upd \Gamma_1 \sa \Delta_1, A \upd
      \mathcal{I}_1}{\mathcal{D}_3}
  }
\]
or
\[
  \infer[R]{\mathcal{G} \upd \Gamma \sa \Delta, A \upd \mathcal{I}
  }
  {\infer*{\mathcal{G} \upd \Gamma_1 \sa \Delta_1, A \upd
      \mathcal{I}_1}{\mathcal{D}_3}
    &
    \infer*{\mathcal{G} \upd \Gamma_2 \sa \Delta_2, A \upd
      \mathcal{I}_2}{\mathcal{D}_4}
  }
\]
for some rule $R$. Applying $\SL(n,m-1)$ to the premiss(es) or $R$ and
$\mathcal{D}_2$ yields
\[
  \mathcal{G} \mrg \mathcal{H} \upd \Gamma_i,\Sigma \sa \Delta_i,\Pi
  \upd \mathcal{I}_i
\]
and applying the rule $R$ gives the desired result. \qed
\end{proof}
\end{document}